\documentclass[a4paper,twocolumn,10pt]{quantumarticle}
\pdfoutput=1
\usepackage{array}
\usepackage{multirow}
\usepackage{amsmath}
\usepackage{amsthm}
\usepackage{amssymb}
\usepackage[numbers,numbers, sort&compress]{natbib}

\makeatletter

\providecommand{\tabularnewline}{\\}

\theoremstyle{plain}
\newtheorem{thm}{\protect\theoremname}
\theoremstyle{plain}
\newtheorem{lem}[thm]{\protect\lemmaname}
\ifx\proof\undefined
\newenvironment{proof}[1][\protect\proofname]{\par
	\normalfont\topsep6\p@\@plus6\p@\relax
	\trivlist
	\itemindent\parindent
	\item[\hskip\labelsep\scshape #1]\ignorespaces
}{%
	\endtrivlist\@endpefalse
}
\providecommand{\proofname}{Proof}
\fi

\usepackage[utf8]{inputenc}
\usepackage[english]{babel}
\usepackage[T1]{fontenc}

\usepackage{tikz}
\usepackage{lipsum}

\usepackage{xcolor}
\usepackage{nameref}

\definecolor{ForestGreen}{rgb}{0.1333,0.5451,0.1333}
\definecolor{DarkRed}{rgb}{0.8,0,0}
\definecolor{Red}{rgb}{1,0,0}

\usepackage{hyperref}
\usepackage{cleveref}

\crefalias{thm}{theorem}

\let\oldlem\lem
\renewcommand{\lem}{
\crefalias{thm}{lemma}
\oldlem
}

\AtBeginDocument{%
\let\ref\Cref
}

\makeatother

\providecommand{\lemmaname}{Lemma}
\providecommand{\theoremname}{Theorem}

\begin{document}
\global\long\def\id{\openone}%

\global\long\def\gate#1{\textsc{\textsc{\textsc{#1}}}}%

\global\long\def\T{\gate T}%

\global\long\def\ket#1{\left|#1\right\rangle }%

\global\long\def\bra#1{\left\langle #1\right|}%

\global\long\def\braket#1#2{\left\langle #1|#2\right\rangle }%

\global\long\def\ketbra#1#2{\ket{#1}\bra{#2}}%

\global\long\def\tr#1#2{\text{Tr}_{#1}\left[#2\right]}%

\global\long\def\norm#1{\left\Vert #1\right\Vert }%

\global\long\def\ghl#1{\left\Vert \text{#1}\right\Vert }%

\global\long\def\comment#1{\text{#1}}%

\global\long\def\vcon{c}%

\global\long\def\vrots{n}%

\global\long\def\U#1{U_{#1}}%

\global\long\def\Uc#1{\mathcal{U}_{#1}}%

\title{Halving the cost of quantum multiplexed rotations}
\author{Guang Hao Low}
\orcid{0000-0002-6934-1052}

\email{guanghao.low@microsoft.com}

\affiliation{Microsoft Quantum, Redmond, WA 98052, USA}
\maketitle
\begin{abstract}
We improve the number of $\gate T$ gates needed for a $b$-bit approximation
of a multiplexed quantum gate with $\vcon$ controls applying $\vrots$
single-qubit arbitrary phase rotations from $4\vrots b+\mathcal{O}(\sqrt{c\vrots b})$
to $2\vrots b+\mathcal{O}(\sqrt{c\vrots b})$, and reduce the number
of qubits needed by up to a factor of two. This generic quantum circuit
primitive is found in many quantum algorithms, and our results roughly
halve the cost of state-of-art electronic structure simulations based
on qubitization of double-factorized or tensor-hypercontracted representations.
We achieve this by extending recent ideas on stochastic compilation
of quantum circuits to classical data and discuss space-time trade-offs
and concentration of measure in its implementation.
\end{abstract}

\section{Introduction}

Many quantum algorithms with an exponential run-time advantage over
their classical counterparts require large numbers of qubits and quantum
gates. Applications at scientifically or industrially interesting
scales include estimating energy levels of molecules with hundreds
of spin-orbitals and electrons \citep{vonBurg2020carbon,Lee2020hypercontraction}
and factoring RSA integers with thousands of bits \citep{Gidney2021Factor}.
Solving these problems require at least that many qubits just to encode
the input, upon which billions to trillions of elementary quantum
gates are applied.

At large scale, quantum computation on noisy physical hardware requires
fault-tolerant quantum gates on logical qubits in a quantum error
correcting code \citep{Shor1996FaultTolerance}. Although Clifford
gates $\left\{ \gate{Hadamard},\textsc{Phase},\textsc{Cnot}\right\} $
can be implemented transversally and thus fault-tolerantly on many
error correcting codes, they must be augmented by non-Clifford gates,
typically the $\gate T$ gate, to achieve universal quantum computation.
As a simultaneous transveral implementation of the $\gate T$ gate
is impossible \citep{Zeng2011Transversality}, fault-tolerant $\gate T$
gates are realized by sophisticated techniques such as magic state
distillation \citep{Bravyi2005Universal} or gauge fixing \citep{Paetznick2013GaugeFixing}
that are orders of magnitude more costly. The total number of $\textsc{T}$
gates is thus a good heuristic for understanding the real-world cost
of fault-tolerant quantum algorithms.

Optimizing the decomposition of arbitrary quantum algorithms into
the fewest number of T gates is NP-hard. Hence algorithms are typically
expressed through higher level subroutines for which optimal decompositions
are known such as table-lookup \citep{Low2018Trading}, arithmetic
operators \citep{Gidey2018Addition}, rotations \citep{Ross2016Optimal},
and Fourier transforms \citep{Nam2020Fourier}.

We focus on the multiplexed rotation gate

\begin{align}
\U{\vec{\theta}} & =\sum_{j=0}^{\vcon-1}\ketbra jj\otimes e^{i2\pi\theta_{j}Z},\label{eq:multiplex}
\end{align}
which is another common subroutine. This subroutine is a large fraction
of the cost of important quantum algorithms such as arbitrary quantum
state preparation and unitary synthesis \citep{shende2006synthesis},
and more recently, state-of-art simulations of electronic structure
\citep{vonBurg2020carbon,Lee2020hypercontraction,Low2016Qubitization}.
As arbitrary unitary quantum algorithms also decompose into products
$V=\prod_{k}e^{i2\pi\theta_{k}Z}\cdot V_{k}$ of rotations interleaved
by some unitaries $\vec{V}$, the intuitive application of \ref{eq:multiplex}
lies in executing a superposition of quantum algorithms
\begin{equation}
\U{\overleftrightarrow{\theta}}=\left(\prod_{k=0}^{\vrots-1}\U{\theta_{k,;}}\cdot V_{k}\right),\label{eq:multiplex_many}
\end{equation}
with potentially very different parameters $\theta_{;,j}$, but otherwise
all with the same structure $\vec{V}$.

\begin{table}
\begin{centering}
\begin{tabular}{|c|c|}
\hline 
Result & Number of $\gate{Toffoli}$ gates\tabularnewline
\hline 
\citep{shende2006synthesis,Kliuchnikov2013synthesis} & $\mathcal{O}\left(\vcon\vrots\log\left(\frac{\vcon\vrots}{\epsilon}\right)\right)$\tabularnewline
\hline 
\citep{Childs2018Towards} & $\mathcal{O}\left(\vcon\vrots\log\left(\frac{\vrots}{\epsilon}\right)\right)$\tabularnewline
\hline 
\citep{vonBurg2020carbon} & $\vrots\left\lceil \log_{2}\left(\frac{\vrots\pi}{\epsilon}\right)\right\rceil +\mathcal{O}\left(\sqrt{\vcon\vrots\log_{2}\left(\frac{\vrots}{\epsilon}\right)}\right)$\tabularnewline
\hline 
This work & $\vrots\left\lceil \frac{1}{2}\log_{2}\left(\frac{\vrots\pi^{2}}{2\epsilon}\right)\right\rceil +\mathcal{O}\left(\sqrt{\vcon\vrots\log_{2}\left(\frac{\vrots}{\epsilon}\right)}\right)$\tabularnewline
\hline 
\end{tabular}
\par\end{centering}
\caption{\label{tab:prior_art}Cost comparison of implementing fault-tolerant
multiplexed rotation \ref{eq:multiplex}. Further variations realizing
a $\protect\gate T$ gate -- qubit tradeoffs are in \ref{tab:component_costs}.}
\end{table}

\begin{table*}[t]
\begin{centering}
\begin{tabular}{|c|c|c|c|c|c|c|c|c|c|c|c|}
\hline 
\multirow{2}{*}{System} & \multirow{2}{*}{Method} & \multirow{2}{*}{$N$} & \multirow{2}{*}{$c/10^{2}$} & \multirow{2}{*}{$M/10^{5}$} & \multirow{2}{*}{$\lambda$} & \multicolumn{3}{c|}{Deterministic truncation} & \multicolumn{3}{c|}{Randomized truncation}\tabularnewline
\cline{7-12} \cline{8-12} \cline{9-12} \cline{10-12} \cline{11-12} \cline{12-12} 
 &  &  &  &  &  & $b$ & $\gate{Toffoli}$ & Qubits & $b$ & $\gate{Toffoli}$ & Qubits\tabularnewline
\hline 
FeMoco \citep{vonBurg2020carbon} & $H_{\mathrm{DF}}$ & $54$ & $240$ & $4.7$ & $2$ & $33$ & $2.3\times10^{10}$ & $3.6\times10^{3}$ & $18$ & $82\%$ & $64\%$\tabularnewline
\hline 
XVIII \citep{vonBurg2020carbon} & $H_{\mathrm{DF}}$ & $56$ & $300$ & $4.7$ & $4$ & $33$ & $2.1\times10^{10}$ & $7.4\times10^{3}$ & $18$ & $71\%$ & $62\%$\tabularnewline
\hline 
FeMoco \citep{Lee2020hypercontraction} & $H_{\mathrm{DF}}$ & $54$ & $130$ & $4.6$ & $4$ & $16$ & $1.0\times10^{10}$ & $3.7\times10^{3}$ & $18$ & $114\%$ & $127\%$\tabularnewline
\hline 
FeMoco \citep{Lee2020hypercontraction} & $H_{\mathrm{THC}}$ & $54$ & $3.5$ & $4.8$ & $2$ & $16$ & $\ensuremath{5.3\times10^{9}}$ & $2.1\times10^{3}$ & $18$ & $108\%$ & $110\%$\tabularnewline
\hline 
\end{tabular}
\par\end{centering}
\caption{\label{tab:chem_comparison}Cost comparison of simulating active site
configurations of a nitrogen fixation enzyme (FeMoco) \citep{Reiher2016Reaction}
and a carbon dioxide fixation catalyst (XVIII) \citep{vonBurg2020carbon}
using a randomized implementation of \ref{eq:multiplex_many}. Note
that the choice of $b=\left\lceil \log_{2}\left(160MN\right)\right\rceil \approx33$
by deterministic truncation rigorously upper bounds the error of simulation
\citep{vonBurg2020carbon} and is a fair comparison with our randomized
scheme. In contrast, the choice of $b=16$ was obtained by heuristics
and would be also be roughly twice as large under provable deterministic
conditions. The $\lambda$ parameter controls the space-time trade-off
of table-lookup as in \ref{tab:component_costs}. Due to spin symmetry
in these examples, the Givens rotations of size $N$ here applies
to $2N$-spin-orbital systems.}
\end{table*}

Our main contribution is roughly halving the $\gate T$ gate cost
of implementing \ref{eq:multiplex} and its occurrences in \ref{eq:multiplex_many}
compared to the best prior methods summarized in \ref{tab:prior_art}.
Given a target error $\epsilon$ in diamond distance, we show in \ref{sec:Multiplexed-rotation-circuit}
that it suffices to use
\begin{align}
\vrots\left\lceil \frac{1}{2}\log_{2}\left(\frac{\vrots\pi^{2}}{2\epsilon}\right)\right\rceil +\mathcal{O}\left(\sqrt{\vcon\vrots\log_{2}\left(\frac{\vrots}{\epsilon}\right)}\right)\label{eq:main_result}
\end{align}
$\gate{Toffoli}$ gates, which we refer to interchangeably as four
$\gate T$ gates and vice-versa \citep{Howard2017MagicStates,Gidey2018Addition}.
Our approach also enables a trade-off between $\gate T$ count and
qubit count. This result is widely applicable. For instance we demonstrate
in \ref{tab:chem_comparison} and \ref{sec:Simulation-of-electronic}
how our results almost half the cost of state-of-art electronic structure
simulation. 

The basic idea exploits classical randomness to approximate real numbers
using a $b/2$-bit randomized truncation scheme instead of deterministic
truncating to $b$-bit as outlined in \ref{sec:Randomization-as-a}.
Though our results are related to and can be proven using more general
machinery on mixing unitaries in \citep{Campbell2017mixing,Hastings2016Incoherent},
our construction uses a very simple randomization that admits a simpler
proof with tighter bounds. Of independent interest in \ref{sec:Concentration-of-measure}
is the concentration of measure of our randomized implementation of
multiplexed rotations.

\section{\label{sec:Randomization-as-a}Randomization as a resource}

The trace distance $\mathcal{D}$ is a good characterization of the
difference between quantum states and channels. For any two quantum
states $\rho$ and $\rho'$, 
\begin{align}
\mathcal{D}(\rho,\rho') & \doteq\frac{1}{2}\norm{\rho-\rho'}_{1},\label{eq:trace_distance}
\end{align}
where $\norm{\cdot}_{1}$ is the Schatten $1$-norm. This equals the
largest possible total variational distance of measurement outcomes.
Suppose $\left\{ E_{m}\right\} $ is a positive operator valued measurement
with corresponding measurement probabilities $p_{m}=\tr{}{\rho E_{m}}$
for outcome $m$, and similarly for $p_{m}'$. The distance between
these distributions is then $\max_{\left\{ E_{m}\right\} }\frac{1}{2}\sum_{m}|p_{m}-p_{m}'|=\mathcal{D}_{\mathrm{tr}}(\rho,\rho')$
\citep{Nielsen2004}. In the case of finite dimensional quantum channels
$\mathcal{V}$ and $\mathcal{V}'$, it is common to use the completely
bounded trace norm, or diamond distance \citep{Kitaev1997Diamond}
\begin{align}
\mathcal{D}_{\diamond}(\mathcal{V},\mathcal{V}') & =\max_{\rho}\mathcal{D}\left(\mathcal{V}\otimes\mathcal{I}\left(\rho\right),\mathcal{V}'\otimes\mathcal{I}\left(\rho\right)\right).\label{eq:diamond_distance}
\end{align}
This expression through the trace distance highlights its operational
meaning. It characterizes the distinguishability of quantum channels
with respect to any quantum state, even allowing for entanglement
assistance.

Although \ref{eq:diamond_distance} is difficult to evaluate in general,
simplifications are known. For instance, convexity of trace distance
implies that maximization over density matrices $\rho$ can be replaced
by that over pure states. Moreover, the dimension of $\mathcal{I}$
may be limited to that of $\mathcal{V}$ without loss of generality.
In fact, when $\mathcal{V}$ and $\mathcal{V}'$ correspond to unitaries
$V$ and $V'$, stabilization by $\mathcal{I}$ is unnecessary \citep{Wang2013Solovay}.
Our work relies on the following special case, which is tighter than
some similar results \citep{Campbell2017mixing,Hastings2016Incoherent}.
\begin{lem}[{\label{lem:Diamond-distance}\citep[Lemmas 2]{Chen2020} Diamond distance
of unitary ensembles}]
Let $V$ be unitary and let the channel $\text{\ensuremath{\mathcal{V}'}}(\cdot)=\mathbb{E}\left[V^{'\dagger}(\cdot)V'\right]$
apply the randomly sampled unitary $V'$. Then
\begin{equation}
\mathcal{D}_{\diamond}(\mathcal{V},\mathcal{V}')\le\norm{V-\mathbb{E}\left[V'\right]}.
\end{equation}
\end{lem}
Randomness is a powerful tool for approximating continuous quantities
with discrete quantities. Consider the deterministic rounding of a
complex phase $e^{i2\pi\theta}$. As the phase has a binary expansion
$\theta=\sum_{l=1}^{\infty}\theta_{l}2^{-l}\in[0,1)$, where $\theta_{l}\in\left\{ 0,1\right\} $,
its $b$-bit approximation $\theta^{(b)}\doteq\sum_{l=1}^{b}\theta_{l}2^{-l}-2^{-b-1}$
has error at most $2^{-b-1}$. Hence choosing $b=\left\lceil \log_{2}\left(\frac{\pi}{\epsilon}\right)\right\rceil $
guarantees

\begin{align}
\left|e^{i2\pi\theta}-e^{i2\pi\theta^{(b)}}\right| & \le2\pi\left(\theta-\theta^{(b)}\right)\le\frac{\pi}{2^{b}}\le\epsilon.\label{eq:deterministic_rounding}
\end{align}
We may in fact halve the bits of precision with a randomized rounding
strategy that will be useful in the following.
\begin{lem}[\label{lem:random_phase_rounding}Randomized rounding of phases]
 Let $X\sim\mathrm{Bernoulli}\left(r\right)$ where $r=2^{b}\left(\theta-\theta^{(b-1)}\right)$,
and let $\Theta^{(b)}\doteq\theta^{(b-1)}+2^{-b}X$ be a random $b$-bit
angle. Then $b=\left\lceil \frac{1}{2}\log_{2}\left(\frac{\pi^{2}}{2\epsilon}\right)\right\rceil $
bits suffice to bound
\begin{equation}
\left|e^{i2\pi\theta}-\mathbb{E}\left[e^{i2\pi\Theta^{(b)}}\right]\right|\le\epsilon.\label{eq:random_phase_approximation}
\end{equation}
\end{lem}
\begin{proof}
The distance between the boundary of a circle and any point on a given
chord is maximized at the halfway mark $r=1/2$. Hence \ref{eq:random_phase_approximation}
is bounded by $\left|1-\frac{e^{-i\pi2^{-b}}+e^{i\pi2^{-b}}}{2}\right|=\left|1-\cos\left(\pi/2^{b}\right)\right|\le\frac{\pi^{2}}{2^{2b+1}}$.
\end{proof}
In other words, the random discrete angles are designed to match the
target angle in expectation $\mathbb{E}\left[\Theta^{(b)}\right]=\theta$,
and the simplest example is the distribution on $\theta^{(b-1)}\pm2^{-b}.$ 

\section{\label{sec:Multiplexed-rotation-circuit}Multiplexed rotation circuit}

The multiplexed rotation gate \ref{eq:multiplex} has a few well-known
implementations. Our approach adds randomization on top of prior state-of-art
that encodes the desired rotation angles as binary integers in a quantum
lookup table \citep{vonBurg2020carbon}. Let us define the table-lookup
unitary as
\begin{equation}
\gate D_{\vcon,b}\ket j\ket z=\ket j\ket{z\oplus\theta_{j}^{(b)}}.\label{eq:lookup}
\end{equation}
Controlled on an integer $j\in[\vcon]$, table-lookup performs a bit-wise
XOR of the $b$-bit angle $\theta_{j}^{(b)}$ into the input register
$\ket z$, typically initialized as $\ket 0^{\otimes b}$. Subsequently,
a controlled rotation 
\begin{equation}
\gate R_{b}\ket{\theta_{j}^{(b)}}\ket{\psi}=\ket{\theta_{j}}e^{i2\pi\theta_{j}^{(b)}Z}\ket{\psi},\label{eq:rotation}
\end{equation}
is performed. With a final step of uncomputing $D$, this realizes
the multiplexed rotation $\U{\vec{\theta}}=\gate D_{\vcon,b}^{\dagger}\gate R_{b}\gate D_{\vcon,b}$
through the sequence
\begin{align}
\ket j\ket{\psi}\ket 0\underset{\gate D_{\vcon,b}}{\rightarrow} & \ket j\ket{\psi}\ket{\theta_{j}^{(b)}}\underset{\gate R_{b}}{\rightarrow}\ket je^{i2\pi\theta_{j}^{(b)}Z}\ket{\psi}\ket{\theta_{j}^{(b)}}\nonumber \\
\underset{\gate D_{\vcon,b}^{\dagger}}{\rightarrow} & \ket je^{i2\pi\theta_{j}^{(b)}Z}\ket{\psi}\ket 0.\label{eq:multiplex_circuit_steps}
\end{align}
The case $\U{\overleftrightarrow{\theta}}$ in \ref{eq:multiplex_many}
may be trivially implemented by applying \ref{eq:multiplex} $\vrots$
times. As the $\gate T$ cost may depends $\vcon$, a more $\gate T$
efficient approach \citep{vonBurg2020carbon} applies $\lceil n/k\rceil$
layers of multiplexed rotation gates each using $\gate D_{\vcon,kb}$.
This uses more space as $kb$ bits are written out in parallel. In
the extreme case, $\gate D_{\vcon,nb}$ may write out all $\vrots$
rotation angles at once.

Whereas prior approaches took a deterministic rounding of $\theta$,
we halve the bits of precision needed by applying a randomized rounding.
Let us approximate $\U{\vec{\theta}}$ by $\U{\vec{\Theta}^{(b)}}$,
where each angle $\theta_{j}$ is replaced by the random $b$-bit
variable $\Theta_{j}^{(b)}$ as described in \ref{lem:random_phase_rounding}.
Whenever $U\left[\vec{\theta}\right]$ is required by a quantum algorithm,
we randomly sample $\vec{\theta}'\sim\vec{\Theta}^{(b)}$and instead
apply $\U{\vec{\theta}'}$. The following theorem bounds the error
of this procedure.
\begin{thm}[\label{thm:randomized_multiplex}Randomized compilation of multiplexed
rotations]
The distance between the channels $\Uc{\vec{\theta}}$ and $\Uc{\vec{\Theta}^{(b)}}$
corresponding to the multiplexed rotation $\U{\vec{\theta}}$ and
its $b=\left\lceil \frac{1}{2}\log_{2}\left(\frac{\pi^{2}}{2\epsilon}\right)\right\rceil $-bit
randomization $\U{\vec{\Theta}^{(b)}}$ is
\begin{equation}
\mathcal{D}_{\diamond}\left(\Uc{\vec{\theta}},\Uc{\vec{\Theta}^{(b)}}\right)\le\epsilon.\label{eq:multiplex_error}
\end{equation}
\end{thm}
\begin{proof}
From \ref{lem:Diamond-distance}, \ref{eq:multiplex_error} is bounded
by $\norm{\Uc{\vec{\theta}}-\mathbb{E}\left[\Uc{\vec{\Theta}^{(b)}}\right]}$,
which equals
\begin{align}
 & \norm{\sum_{j=0}^{n-1}\ketbra jj\otimes\left(e^{i2\pi\theta_{j}Z}-\mathbb{E}\left[e^{i2\pi\Theta_{j}^{(b)}Z}\right]\right)}\nonumber \\
 & \le\max_{j}\norm{e^{i2\pi\theta_{j}Z}-\mathbb{E}\left[e^{i2\pi\Theta_{j}^{(b)}Z}\right]}\nonumber \\
 & \le\max_{j,\theta}\left|e^{i2\pi\theta}-\mathbb{E}\left[e^{i2\pi\Theta^{(b)}}\right]\right|\le\frac{\pi^{2}}{2^{2b+1}}.
\end{align}
We then apply \ref{lem:random_phase_rounding} to the last line.
\end{proof}
Some useful variations of $\Uc{\vec{\theta}}$ admit further optimization.
For instance, a triangle inequality bounds the error of $\mathcal{D}_{\diamond}\left(\Uc{\overleftrightarrow{\theta}},\Uc{\overleftrightarrow{\Theta}^{(b)}}\right)\le\frac{\vrots\pi^{2}}{2^{2b+1}}\le\epsilon$
by choosing $b=\left\lceil \frac{1}{2}\log_{2}\left(\frac{\vrots\pi^{2}}{2\epsilon}\right)\right\rceil $.
Furthermore, the same set of angles angle $\vec{\theta}$ may be used
multiple times such as in 
\begin{equation}
\U{\overleftrightarrow{\theta}}=\left(\prod_{k=0}^{\vrots-1}\U{\vec{\theta}}\cdot V_{k}\right).\label{eq:same_angles}
\end{equation}
Whereas the $\gate D_{\vcon,\vrots b}$ implementation of $\Uc{\overleftrightarrow{\Theta}^{(b)}}$
writes out $n$ rotation angles, these angles are repeated and only
differ in the least significant bit. Thus there are only $n+b-1$
unique bits and it suffices to instead use the significantly cheaper
table-lookup $\gate D_{\vcon,\vrots+b-1}$.

\subsection{Cost}

The multiplexed rotation has a gate cost corresponding to the sum
of its components \ref{eq:multiplex_circuit_steps} in \ref{tab:component_costs},
and an ancillae overhead corresponding to the maximum required across
all steps.

\begin{table}
\begin{centering}
\begin{tabular}{|c|c|c|}
\hline 
 & $\gate{Toffoli}$ & ancillae\tabularnewline
\hline 
$\gate D_{\vcon,b}$ & $\left\lceil \frac{\vcon}{\lambda}\right\rceil +(\lambda-1)b$ & $\left\lceil \log_{2}\frac{\vcon}{\lambda}\right\rceil +(\lambda-1)b$\tabularnewline
\hline 
$\gate D_{\vcon,b}^{\dagger}$ & $\left\lceil \frac{\vcon}{\lambda'}\right\rceil +\lambda'$ & $\left\lceil \log_{2}\frac{\vcon}{\lambda'}\right\rceil +\lambda'$\tabularnewline
\hline 
$\textsc{R}_{b}$ & $b$ & $b$\tabularnewline
\hline 
\end{tabular}
\par\end{centering}
\caption{\label{tab:component_costs}Cost of multiplexed rotation components
including table-lookup \ref{eq:lookup} and uncomputation parameterized
by $\lambda,\lambda'$ realizing a space-time trade-off, and the controlled
rotation \ref{eq:rotation}.}
\end{table}
The cost of table-lookup $\gate D_{\vcon,b}$ \citep{Low2018Trading,Berry2019CholeskyQubitization}
is parameterized by a power of $2$ integer $1\le\lambda\le c$ that
governs a trade-off between number of $\gate T$ gates and ancillary
qubits used. Moreover, the power of $2$ constraint on $\lambda$
can be dropped using modular division \citep{vonBurg2020carbon} costing
$\mathcal{O}\left(\log\vcon\right)$ $\gate T$ gates and $\mathcal{O}(1)$
ancillae. The minimum $\gate{Toffoli}$ count is thus $\approx2\sqrt{cb}$
at $\lambda\approx\sqrt{c/b}$ using $\approx\sqrt{cb}$ ancillae.
Nevertheless, a sub-optimal $\lambda$ may still be chosen to limit
the number of ancillae used. 

Importantly, all these ancillae are immediately reset to the $\ket 0$
state and may be reused. This is through a measurement-based uncomputation
process also parameterized by a power of $2$ integer $1\le\lambda'\le c$.
This uncomputation step has a minimum $\gate{Toffoli}$ cost of $\approx2\sqrt{c}$,
and will tend to have negligible cost compared to $\left(\vcon,b\right)$-lookup
whenever $\lambda,b>1$ as we may choose any $\lambda'\le(\lambda-1)b$.

The controlled rotation \ref{eq:rotation} is implemented by the phase
gradient technique \citep{Gidey2018Addition,Lee2020hypercontraction}
costing $b$ Toffoli gates. Given a phase state $\ket{\phi}=\frac{1}{\sqrt{2^{b+1}}}\sum_{k=1}^{2^{b+1}}e^{-i2\pi k/2^{b+1}}\ket k$
and a $b+1$-bit reversible adder $\gate{Add}$, observe that adding
the $b+1$ bit integer $\ket l$ kicks back a phase
\begin{align}
\gate{Add}\ket l\ket{\phi} & =e^{i2\pi l/2^{b+1}}\ket l\ket{\phi}.
\end{align}
As the phase state catalyzes the operation and may be reused, we ignore
its one-time cost of $\mathcal{O}\left(b\log(b/\epsilon)\right)$
$\gate T$ gates.

Hence the controlled adder $\textsc{Add}'=\gate{Add}\otimes\ketbra 01+\gate{Add}^{\dagger}\otimes\ketbra 11$
performs a $Z$ rotation as
\begin{equation}
\textsc{Add}'\ket l\ket{\phi}\ket x=e^{(-1)^{x}i2\pi l/2^{b+1}}\ket l\ket{\phi}\ket x.\label{eq:controlled_adding}
\end{equation}
Note that $\textsc{Add}'$ has the same Toffoli cost as $\textsc{Add}$
as an adder can be converted into a subtractor using only Clifford
gates. This follows from applying a bit-wise complement using $\textsc{Cnot}$
gates in the identity $a-b=\overline{\bar{a}+b}$. Comparing \ref{eq:rotation}
to \ref{eq:controlled_adding}, $\textsc{R}_{b}$ is thus implemented
by fixing the least significant bit of $l$ to $1$ and the remaining
$b$ bits to $\theta^{(b)}$. 

The overall $\gate{Toffoli}$ cost of \ref{eq:multiplex_many} through
$\lceil n/k\rceil$ layers of $\gate D_{\vcon,kb}$-lookup is then
\begin{equation}
\vrots b+\left\lceil \frac{\vrots}{k}\right\rceil \left(\left\lceil \frac{\vcon}{\lambda}\right\rceil +\left\lceil \frac{\vcon}{\lambda'}\right\rceil +(\lambda-1)kb+\lambda'\right).
\end{equation}
Some important special cases of the $\gate{Toffoli}$ count for parameter
choices $(k,\lambda,\lambda')$ include
\begin{align}
(\vrots,1,1) & :\vrots b+2\vcon+1,\nonumber \\
(1,1,1) & :\vrots(b+2\vcon+1),\nonumber \\
(k^{*},\lambda^{*},\lambda^{'*}) & :\vrots b+2\sqrt{\vcon}+2\sqrt{\vcon nb}.
\end{align}
Minimizing this by choosing $k^{*}\sim n,\lambda^{*}\sim\sqrt{c/nb},\lambda^{'*}\sim\sqrt{c}$
recovers \ref{eq:main_result}. 

\section{\label{sec:Simulation-of-electronic}Simulation of electronic structure}

Our work reduces the $\gate{Toffoli}$ and qubit count of previous
state-of-art on precisely estimating energy levels of active site
configurations of a nitrogen fixation enzyme (FeMoco) \citep{Reiher2016Reaction}
and a carbon dioxide fixation catalyst (XVIII) \citep{vonBurg2020carbon}
using the quantum phase estimation algorithm wrapped around queries
to quantum walk unitaries. The percentage reduction in the space-time
$\gate{Toffoli}$-qubit product is roughly $50\%$, as highlighted
in \ref{tab:chem_comparison}. We obtained these numbers by subtracting
the cost of deterministic multiplexed rotations from previous estimates
\citep{vonBurg2020carbon,Lee2020hypercontraction} and adding in the
cost of our randomized approach. 

In previous deterministic approaches, a very small approximation error
$\epsilon\sim10^{-3}$ is chosen for these examples in order to bound
the systematic shift in eigenvalues of each quantum walk. Quantum
phase estimation to error $\delta$ then estimates the energy with
an error of $\pm(\epsilon+\delta)$ with high probability $p$. In
our analysis, we instead chose a target diamond distance $\epsilon=0.05$,
which implies that the estimated energy will be correct to error $\delta$
with at least a confidence of $p-\epsilon$. Thus the cost of phase
estimation to the same overall error using our approach can be further
multiplied by a factor of $\frac{\delta}{\epsilon+\delta}\frac{p}{p-\epsilon},$
which is generally less than one for well-supported trial states.
However, we omit this factor in the table to focus on changes in cost
due to just changes in bits of precision.

In such electronic structure simulation problems, the goal is to accurately
estimate an eigenvalue of the 'double-factorized' Hermitian operator,
or Hamiltonian,
\begin{equation}
H=H_{h^{(0)}}+\sum_{r=1}^{R}H_{h^{(r)}}^{2},\;H_{h}\doteq\sum_{p,q=1}^{N}h_{pq}P_{p,0}P_{q,1},\label{eq:double_factorized}
\end{equation}
where the $\left\{ P_{j,x}\right\} $ are mutually anti-commuting
Pauli operations and $h^{(r)}$ are anti-Hermitian matrices \citep{vonBurg2020carbon}.
An alternate 'tensor-hypercontracted' representation \citep{Lee2020hypercontraction}
is
\begin{equation}
H=H=H_{h^{(0)}}+\sum_{r,s=1}^{R}t_{rs}H_{h^{(r)}}H_{h^{(s)}},\label{eq:THC}
\end{equation}
where $\left\{ h^{(r)}\right\} _{r=1}^{R}$ are additionally rank-$1$
and $t$ is a symmetric matrix. These are all equivalent to Fermion
Hamiltonians with one- and two-body number-conserving interactions.
The basic approach applies quantum phase estimation \citep{Nielsen2004}
on a quantum walk $W$ encoding the spectrum of $H$ \citep{Low2016Qubitization}.
Given a trial state $\ket{\psi}$, this estimates an eigenphase $W\ket{\lambda_{j}}=e^{i\phi_{j}}\ket{\lambda_{j}}$
where the index $j$ is sampled with probability $p_{j}=\left|\braket{\lambda_{j}}{\psi}\right|^{2}$,
and $\phi_{j}$ is estimated to a chosen error $\delta$ with high
probability. Within the algorithm, the most expensive component is
the sequence of the controlled-walk operators applied $M=\mathcal{O}(1/\delta)$
times.

Multiplexed rotations are a large fraction of the cost of these walk
operators in state-of-art quantum resource estimates on useful instances
of this problem \citep{vonBurg2020carbon,Lee2020hypercontraction}.
These occur through so-called controlled Givens rotation unitaries
$G=\sum_{j=0}^{\vcon-1}\ketbra jj\otimes G_{j}$ that diagonalize
the quadratic Hamiltonian $\sum_{pq}h_{pq}P_{p,0}P_{q,1}=\sum_{p\in[\Xi]}\lambda_{p}G_{p}^{\dagger}P_{0,0}P_{0,1}G_{p}$.
In \ref{eq:double_factorized}, we see that $\vcon_{\mathrm{DF}}=(R+1)\Xi$
Givens rotations are required to diagonalize its $R+1$ quadratic
Hamiltonians. Similarly \ref{eq:THC} requires $\vcon_{\mathrm{THC}}=R+N$
Givens rotations. The linear-combination-unitaries technique that
encodes $H$ into a quantum walk applies the sequence 
\begin{equation}
G\left(\id\otimes\cdots\right)G{}^{\dagger}\cdots G\left(\id\otimes\cdots\right)G^{\dagger},
\end{equation}
where no operation is applied on the $\ket j$ register within each
round bracket. As each Givens rotation decomposes into 
\begin{equation}
G_{j}=\left(\prod_{k=0}^{N-2}V_{k,2}e^{i2\pi\theta_{k,j}Z}V_{k,1}e^{i2\pi\theta_{k,j}Z}V_{k,0}\right),
\end{equation}
each half $G\left(\id\otimes\cdots\right)G{}^{\dagger}$ is a multiplexed
gate with $\vcon$ controls and $4(N-1)$ rotations, which matches
the form of \ref{eq:multiplex_many}. Naively, each multiplexed rotation
would require table-lookup $\gate D_{\vcon,4(N-1)b}$ outputting $4(N-1)$
bits. However, each rotation angle is repeated $4$ times, and following
\ref{eq:same_angles} it suffices to randomize only their least significant
bit. Hence it suffices to use roughly $4$ times fewer bits through
$\gate D_{\vcon,(N-1)(b-1)+4(N-1)}$. The diamond distance of $2M$
randomized applications of $G\left(\id\otimes\cdots\right)G{}^{\dagger}$
in quantum phase estimation is then from \ref{thm:randomized_multiplex}
at most $\frac{8M(N-1)\pi^{2}}{2^{2b+1}}\le\epsilon$ with the choice
\begin{equation}
b=\left\lceil \frac{1}{2}\log_{2}\left(\frac{4M(N-1)\pi^{2}}{\epsilon}\right)\right\rceil .
\end{equation}

\section{\label{sec:Concentration-of-measure}Concentration of measure}

We have focused on the diamond distance in bounding the error of our
randomized protocol. However, one might wonder about the error of
a given sampled instance of $\mathcal{V}'$ rather than the channel
average, that is, the quantity
\begin{equation}
\norm{V'-V}=\max_{\ket{\psi}}\norm{(V'-V)\ket{\psi}}.
\end{equation}
We may also reason about the average error by, for instance, bounding
$\mathbb{E}\left[\max_{\ket{\psi}}\norm{(V'-V)\ket{\psi}}\right]$.
However, this average is too pessimistic. It states that for any sampled
$V'$, we adversarially choose the worst $\ket{\psi}$. Subsequently,
we take the average error. In other words, this $V'$ is expected
to perform well for all possible input states. In practice however,
a sampled $V'$ is only applied once to a single input state. Future
repeats would re-sample $V'$ rather than reuse the original sample.
Hence the more relevant measure of error is
\begin{equation}
\mathbb{\max_{\ket{\psi}}E}\left[\norm{(V'-V)\ket{\psi}}\right].
\end{equation}

It is even more informative to bound the probability of large deviations
of the error from the mean
\begin{equation}
\max_{\ket{\psi}}\mathrm{Pr}\left[\norm{(V'-V)\ket{\psi}}\ge\epsilon\right].
\end{equation}
Closely following the derivation by Chen et al. \citep{Chen2020},
we prove in the appendix that that the error of a random instance
$\norm{\U{\overleftrightarrow{\theta}}-\U{\overleftrightarrow{\Theta}^{(b)}}}$
is small with high probability. In other words, 
\begin{equation}
\max_{\ket{\psi}}\mathbb{E}\left[\norm{\left(\U{\overleftrightarrow{\theta}}-\U{\overleftrightarrow{\Theta}^{(b)}}\right)\ket{\psi}}\right]\lesssim4\pi\frac{\sqrt{\vrots}}{2^{b}}.
\end{equation}
Note that the dependence on $2^{b}$ is quadratically worse than the
quantum channel analysis. In more detail, we establish tail bounds
on the error distribution
\begin{align}
\max_{\ket{\psi}}\mathrm{Pr}\left[\norm{\left(\U{\overleftrightarrow{\theta}}-\U{\overleftrightarrow{\Theta}^{(b)}}\right)\ket{\psi}}\ge\epsilon\right] & \le\exp\left(-\frac{\epsilon^{2}2^{2b}}{32e\pi^{2}\vrots}\right).
\end{align}
The bits of precision required is then 
\begin{align}
b & \ge4.88+\log_{2}\left(\log\left(\frac{1}{p}\right)\right)+\log_{2}\left(\frac{\sqrt{\vrots}}{\epsilon}\right).
\end{align}

Though interesting, any useful quantum computation terminates with
a measurement of the quantum state. In the context of measurement,
these single-shot error bounds are thus less useful as the diamond
distance directly characterizes the deviation of measurement outcomes
and is also tighter.

\section{Conclusion}

We have presented a very simple randomized protocol for implementing
the common quantum circuit subroutine of a single, or sequence of,
multiplexed rotations. Our implementation uses roughly half the $\T$
gates and half the qubits of prior deterministic approaches. Whereas
prior methods of stochastic compilation focus on approximating arbitrary
single-qubit rotations with samples from a dense $\epsilon$-mesh
of single-qubit rotations, we focus on stochastic compilation of more
complex multi-qubit operations in the context of directly minimizing
$\gate T$ count. Moreover, our randomization protocol, based only
on flipping biased coins, is extremely simple. We have evaluated the
relevance of our results in halving the space-time cost of important
applications such as electronic structure simulation based on qubitization,
and our work elucidate a path towards systematically applying stochastic
compilation methods beyond \citep{Rajput2021hybridizedsimulation}
Trotter-based methods of simulation \citep{Campbell2019Random}, such
as quantum signal processing \citep{Novo2020randomperspective,Low2016HamSim},
which has proven to be surprisingly elusive.

\onecolumn

\bibliographystyle{ACM-Reference-Format}
\bibliography{export}


\begin{thebibliography}{28}


\ifx \showCODEN    \undefined \def \showCODEN     #1{\unskip}     \fi
\ifx \showDOI      \undefined \def \showDOI       #1{#1}\fi
\ifx \showISBNx    \undefined \def \showISBNx     #1{\unskip}     \fi
\ifx \showISBNxiii \undefined \def \showISBNxiii  #1{\unskip}     \fi
\ifx \showISSN     \undefined \def \showISSN      #1{\unskip}     \fi
\ifx \showLCCN     \undefined \def \showLCCN      #1{\unskip}     \fi
\ifx \shownote     \undefined \def \shownote      #1{#1}          \fi
\ifx \showarticletitle \undefined \def \showarticletitle #1{#1}   \fi
\ifx \showURL      \undefined \def \showURL       {\relax}        \fi
\providecommand\bibfield[2]{#2}
\providecommand\bibinfo[2]{#2}
\providecommand\natexlab[1]{#1}
\providecommand\showeprint[2][]{arXiv:#2}

\bibitem[\protect\citeauthoryear{Berry, Gidney, Motta, McClean, and
  Babbush}{Berry et~al\mbox{.}}{2019}]%
        {Berry2019CholeskyQubitization}
\bibfield{author}{\bibinfo{person}{Dominic~W. Berry}, \bibinfo{person}{Craig
  Gidney}, \bibinfo{person}{Mario Motta}, \bibinfo{person}{Jarrod~R. McClean},
  {and} \bibinfo{person}{Ryan Babbush}.} \bibinfo{year}{2019}\natexlab{}.
\newblock \showarticletitle{Qubitization of Arbitrary Basis Quantum Chemistry
  Leveraging Sparsity and Low Rank Factorization}.
\newblock \bibinfo{journal}{\emph{arXiv preprint arXiv:1902.02134}}
  (\bibinfo{date}{2} \bibinfo{year}{2019}).
\newblock
\urldef\tempurl%
\url{http://arxiv.org/abs/1902.02134}
\showURL{%
\tempurl}


\bibitem[\protect\citeauthoryear{Bravyi and Kitaev}{Bravyi and Kitaev}{2005}]%
        {Bravyi2005Universal}
\bibfield{author}{\bibinfo{person}{Sergey Bravyi} {and} \bibinfo{person}{Alexei
  Kitaev}.} \bibinfo{year}{2005}\natexlab{}.
\newblock \showarticletitle{Universal quantum computation with ideal Clifford
  gates and noisy ancillas}.
\newblock \bibinfo{journal}{\emph{Physical Review A}}  \bibinfo{volume}{71}
  (\bibinfo{date}{2} \bibinfo{year}{2005}), \bibinfo{pages}{022316}.
\newblock
Issue 2.
\showISSN{1050-2947}
\urldef\tempurl%
\url{https://doi.org/10.1103/PhysRevA.71.022316}
\showDOI{\tempurl}


\bibitem[\protect\citeauthoryear{Campbell}{Campbell}{2017}]%
        {Campbell2017mixing}
\bibfield{author}{\bibinfo{person}{Earl Campbell}.}
  \bibinfo{year}{2017}\natexlab{}.
\newblock \showarticletitle{Shorter gate sequences for quantum computing by
  mixing unitaries}.
\newblock \bibinfo{journal}{\emph{Physical Review A}}  \bibinfo{volume}{95}
  (\bibinfo{date}{4} \bibinfo{year}{2017}), \bibinfo{pages}{042306}.
\newblock
Issue 4.
\showISSN{2469-9926}
\urldef\tempurl%
\url{https://doi.org/10.1103/PhysRevA.95.042306}
\showDOI{\tempurl}


\bibitem[\protect\citeauthoryear{Campbell}{Campbell}{2019}]%
        {Campbell2019Random}
\bibfield{author}{\bibinfo{person}{Earl Campbell}.}
  \bibinfo{year}{2019}\natexlab{}.
\newblock \showarticletitle{A random compiler for fast Hamiltonian simulation}.
\newblock \bibinfo{journal}{\emph{Physical Review Letters}}
  \bibinfo{volume}{(Accepted)} (\bibinfo{date}{11} \bibinfo{year}{2019}).
\newblock
\urldef\tempurl%
\url{http://arxiv.org/abs/1811.08017}
\showURL{%
\tempurl}


\bibitem[\protect\citeauthoryear{Chen, Huang, Kueng, and Tropp}{Chen
  et~al\mbox{.}}{2021}]%
        {Chen2020}
\bibfield{author}{\bibinfo{person}{Chi-Fang Chen}, \bibinfo{person}{Hsin-Yuan
  Huang}, \bibinfo{person}{Richard Kueng}, {and} \bibinfo{person}{Joel~A.
  Tropp}.} \bibinfo{year}{2021}\natexlab{}.
\newblock \showarticletitle{Concentration for Random Product Formulas}.
\newblock \bibinfo{journal}{\emph{PRX Quantum}}  \bibinfo{volume}{2}
  (\bibinfo{date}{10} \bibinfo{year}{2021}).
\newblock
Issue 4.
\showISSN{2691-3399}
\urldef\tempurl%
\url{https://doi.org/10.1103/PRXQuantum.2.040305}
\showDOI{\tempurl}


\bibitem[\protect\citeauthoryear{Childs, Maslov, Nam, Ross, and Su}{Childs
  et~al\mbox{.}}{2018}]%
        {Childs2018Towards}
\bibfield{author}{\bibinfo{person}{Andrew~M. Childs}, \bibinfo{person}{Dmitri
  Maslov}, \bibinfo{person}{Yunseong Nam}, \bibinfo{person}{Neil~J. Ross},
  {and} \bibinfo{person}{Yuan Su}.} \bibinfo{year}{2018}\natexlab{}.
\newblock \showarticletitle{Toward the first quantum simulation with quantum
  speedup}.
\newblock \bibinfo{journal}{\emph{Proceedings of the National Academy of
  Sciences}}  \bibinfo{volume}{115} (\bibinfo{date}{9} \bibinfo{year}{2018}),
  \bibinfo{pages}{9456--9461}.
\newblock
Issue 38.
\showISSN{0027-8424}
\urldef\tempurl%
\url{https://doi.org/10.1073/pnas.1801723115}
\showDOI{\tempurl}


\bibitem[\protect\citeauthoryear{Gidney}{Gidney}{2018}]%
        {Gidey2018Addition}
\bibfield{author}{\bibinfo{person}{Craig Gidney}.}
  \bibinfo{year}{2018}\natexlab{}.
\newblock \showarticletitle{Halving the cost of quantum addition}.
\newblock \bibinfo{journal}{\emph{Quantum}}  \bibinfo{volume}{2}
  (\bibinfo{date}{6} \bibinfo{year}{2018}), \bibinfo{pages}{74}.
\newblock
\showISSN{2521-327X}
\urldef\tempurl%
\url{https://doi.org/10.22331/q-2018-06-18-74}
\showDOI{\tempurl}


\bibitem[\protect\citeauthoryear{Gidney and Ekerå}{Gidney and Ekerå}{2021}]%
        {Gidney2021Factor}
\bibfield{author}{\bibinfo{person}{Craig Gidney} {and} \bibinfo{person}{Martin
  Ekerå}.} \bibinfo{year}{2021}\natexlab{}.
\newblock \showarticletitle{How to factor 2048 bit RSA integers in 8 hours
  using 20 million noisy qubits}.
\newblock \bibinfo{journal}{\emph{Quantum}}  \bibinfo{volume}{5}
  (\bibinfo{date}{4} \bibinfo{year}{2021}), \bibinfo{pages}{433}.
\newblock
\showISSN{2521-327X}
\urldef\tempurl%
\url{https://doi.org/10.22331/q-2021-04-15-433}
\showDOI{\tempurl}


\bibitem[\protect\citeauthoryear{Hastings}{Hastings}{2016}]%
        {Hastings2016Incoherent}
\bibfield{author}{\bibinfo{person}{M.~B. Hastings}.}
  \bibinfo{year}{2016}\natexlab{}.
\newblock \showarticletitle{Turning Gate Synthesis Errors into Incoherent
  Errors}.
\newblock \bibinfo{journal}{\emph{arXiv preprint arXiv:1612.01011}}
  (\bibinfo{date}{12} \bibinfo{year}{2016}).
\newblock
\urldef\tempurl%
\url{http://arxiv.org/abs/1612.01011}
\showURL{%
\tempurl}


\bibitem[\protect\citeauthoryear{Howard and Campbell}{Howard and
  Campbell}{2017}]%
        {Howard2017MagicStates}
\bibfield{author}{\bibinfo{person}{Mark Howard} {and} \bibinfo{person}{Earl
  Campbell}.} \bibinfo{year}{2017}\natexlab{}.
\newblock \showarticletitle{Application of a Resource Theory for Magic States
  to Fault-Tolerant Quantum Computing}.
\newblock \bibinfo{journal}{\emph{Physical Review Letters}}
  \bibinfo{volume}{118} (\bibinfo{date}{3} \bibinfo{year}{2017}).
\newblock
Issue 9.
\showISSN{0031-9007}
\urldef\tempurl%
\url{https://doi.org/10.1103/PhysRevLett.118.090501}
\showDOI{\tempurl}


\bibitem[\protect\citeauthoryear{Kitaev}{Kitaev}{1997}]%
        {Kitaev1997Diamond}
\bibfield{author}{\bibinfo{person}{A~Yu Kitaev}.}
  \bibinfo{year}{1997}\natexlab{}.
\newblock \showarticletitle{Quantum computations: algorithms and error
  correction}.
\newblock \bibinfo{journal}{\emph{Russian Mathematical Surveys}}
  \bibinfo{volume}{52} (\bibinfo{date}{12} \bibinfo{year}{1997}),
  \bibinfo{pages}{1191--1249}.
\newblock
Issue 6.
\showISSN{0036-0279}
\urldef\tempurl%
\url{https://doi.org/10.1070/RM1997v052n06ABEH002155}
\showDOI{\tempurl}


\bibitem[\protect\citeauthoryear{Kliuchnikov, Maslov, and Mosca}{Kliuchnikov
  et~al\mbox{.}}{2013}]%
        {Kliuchnikov2013synthesis}
\bibfield{author}{\bibinfo{person}{Vadym Kliuchnikov}, \bibinfo{person}{Dmitri
  Maslov}, {and} \bibinfo{person}{Michele Mosca}.}
  \bibinfo{year}{2013}\natexlab{}.
\newblock \showarticletitle{Fast and Efficient Exact Synthesis of Single-Qubit
  Unitaries Generated by Clifford and T Gates}.
\newblock \bibinfo{journal}{\emph{Quantum Info. Comput.}}  \bibinfo{volume}{13}
  (\bibinfo{date}{7} \bibinfo{year}{2013}), \bibinfo{pages}{607–630}.
\newblock
Issue 7–8.
\showISSN{1533-7146}


\bibitem[\protect\citeauthoryear{Lee, Berry, Gidney, Huggins, McClean, Wiebe,
  and Babbush}{Lee et~al\mbox{.}}{2021}]%
        {Lee2020hypercontraction}
\bibfield{author}{\bibinfo{person}{Joonho Lee}, \bibinfo{person}{Dominic~W.
  Berry}, \bibinfo{person}{Craig Gidney}, \bibinfo{person}{William~J. Huggins},
  \bibinfo{person}{Jarrod~R. McClean}, \bibinfo{person}{Nathan Wiebe}, {and}
  \bibinfo{person}{Ryan Babbush}.} \bibinfo{year}{2021}\natexlab{}.
\newblock \showarticletitle{Even More Efficient Quantum Computations of
  Chemistry Through Tensor Hypercontraction}.
\newblock \bibinfo{journal}{\emph{PRX Quantum}}  \bibinfo{volume}{2}
  (\bibinfo{date}{7} \bibinfo{year}{2021}), \bibinfo{pages}{030305}.
\newblock
Issue 3.
\showISSN{2691-3399}
\urldef\tempurl%
\url{https://doi.org/10.1103/PRXQuantum.2.030305}
\showDOI{\tempurl}


\bibitem[\protect\citeauthoryear{Low and Chuang}{Low and Chuang}{2017}]%
        {Low2016HamSim}
\bibfield{author}{\bibinfo{person}{Guang~Hao Low} {and}
  \bibinfo{person}{Isaac~L. Chuang}.} \bibinfo{year}{2017}\natexlab{}.
\newblock \showarticletitle{Optimal Hamiltonian Simulation by Quantum Signal
  Processing}.
\newblock \bibinfo{journal}{\emph{Physical Review Letters}}
  \bibinfo{volume}{118} (\bibinfo{date}{1} \bibinfo{year}{2017}),
  \bibinfo{pages}{010501}.
\newblock
Issue 1.
\showISSN{0031-9007}
\urldef\tempurl%
\url{https://doi.org/10.1103/PhysRevLett.118.010501}
\showDOI{\tempurl}


\bibitem[\protect\citeauthoryear{Low and Chuang}{Low and Chuang}{2019}]%
        {Low2016Qubitization}
\bibfield{author}{\bibinfo{person}{Guang~Hao Low} {and}
  \bibinfo{person}{Isaac~L. Chuang}.} \bibinfo{year}{2019}\natexlab{}.
\newblock \showarticletitle{Hamiltonian simulation by qubitization}.
\newblock \bibinfo{journal}{\emph{Quantum}}  \bibinfo{volume}{3}
  (\bibinfo{year}{2019}), \bibinfo{pages}{163}.
\newblock
\urldef\tempurl%
\url{https://doi.org/doi.org/10.22331/q-2019-07-12-163}
\showDOI{\tempurl}


\bibitem[\protect\citeauthoryear{Low, Kliuchnikov, and Schaeffer}{Low
  et~al\mbox{.}}{2018}]%
        {Low2018Trading}
\bibfield{author}{\bibinfo{person}{Guang~Hao Low}, \bibinfo{person}{Vadym
  Kliuchnikov}, {and} \bibinfo{person}{Luke Schaeffer}.}
  \bibinfo{year}{2018}\natexlab{}.
\newblock \showarticletitle{Trading T-gates for dirty qubits in state
  preparation and unitary synthesis}.
\newblock \bibinfo{journal}{\emph{arXiv preprint arXiv:1812.00954}}
  (\bibinfo{date}{12} \bibinfo{year}{2018}).
\newblock
\urldef\tempurl%
\url{http://arxiv.org/abs/1812.00954}
\showURL{%
\tempurl}


\bibitem[\protect\citeauthoryear{Nam, Su, and Maslov}{Nam
  et~al\mbox{.}}{2020}]%
        {Nam2020Fourier}
\bibfield{author}{\bibinfo{person}{Yunseong Nam}, \bibinfo{person}{Yuan Su},
  {and} \bibinfo{person}{Dmitri Maslov}.} \bibinfo{year}{2020}\natexlab{}.
\newblock \showarticletitle{Approximate quantum Fourier transform with O(n
  log(n)) T gates}.
\newblock \bibinfo{journal}{\emph{npj Quantum Information}}
  \bibinfo{volume}{6} (\bibinfo{date}{12} \bibinfo{year}{2020}).
\newblock
Issue 1.
\showISSN{2056-6387}
\urldef\tempurl%
\url{https://doi.org/10.1038/s41534-020-0257-5}
\showDOI{\tempurl}


\bibitem[\protect\citeauthoryear{Nielsen and Chuang}{Nielsen and
  Chuang}{2004}]%
        {Nielsen2004}
\bibfield{author}{\bibinfo{person}{Michael~A. Nielsen} {and}
  \bibinfo{person}{Isaac~L. Chuang}.} \bibinfo{year}{2004}\natexlab{}.
\newblock \bibinfo{booktitle}{\emph{Quantum computation and quantum
  information} (\bibinfo{edition}{1} ed.)}.
\newblock \bibinfo{publisher}{Cambridge University Press}.
\newblock


\bibitem[\protect\citeauthoryear{Novo}{Novo}{2020}]%
        {Novo2020randomperspective}
\bibfield{author}{\bibinfo{person}{Leonardo Novo}.}
  \bibinfo{year}{2020}\natexlab{}.
\newblock \showarticletitle{Bridging gaps between random approaches to quantum
  simulation}.
\newblock \bibinfo{journal}{\emph{Quantum Views}}  \bibinfo{volume}{4}
  (\bibinfo{date}{3} \bibinfo{year}{2020}).
\newblock
\urldef\tempurl%
\url{https://doi.org/10.22331/qv-2020-03-19-33}
\showDOI{\tempurl}


\bibitem[\protect\citeauthoryear{Paetznick and Reichardt}{Paetznick and
  Reichardt}{2013}]%
        {Paetznick2013GaugeFixing}
\bibfield{author}{\bibinfo{person}{Adam Paetznick} {and}
  \bibinfo{person}{Ben~W. Reichardt}.} \bibinfo{year}{2013}\natexlab{}.
\newblock \showarticletitle{Universal Fault-Tolerant Quantum Computation with
  Only Transversal Gates and Error Correction}.
\newblock \bibinfo{journal}{\emph{Physical Review Letters}}
  \bibinfo{volume}{111} (\bibinfo{date}{8} \bibinfo{year}{2013}).
\newblock
Issue 9.
\showISSN{0031-9007}
\urldef\tempurl%
\url{https://doi.org/10.1103/PhysRevLett.111.090505}
\showDOI{\tempurl}


\bibitem[\protect\citeauthoryear{Rajput, Roggero, and Wiebe}{Rajput
  et~al\mbox{.}}{2021}]%
        {Rajput2021hybridizedsimulation}
\bibfield{author}{\bibinfo{person}{Abhishek Rajput},
  \bibinfo{person}{Alessandro Roggero}, {and} \bibinfo{person}{Nathan Wiebe}.}
  \bibinfo{year}{2021}\natexlab{}.
\newblock \showarticletitle{Hybridized Methods for Quantum Simulation in the
  Interaction Picture}.
\newblock  (\bibinfo{date}{9} \bibinfo{year}{2021}).
\newblock


\bibitem[\protect\citeauthoryear{Reiher, Wiebe, Svore, Wecker, and
  Troyer}{Reiher et~al\mbox{.}}{2017}]%
        {Reiher2016Reaction}
\bibfield{author}{\bibinfo{person}{Markus Reiher}, \bibinfo{person}{Nathan
  Wiebe}, \bibinfo{person}{Krysta~M Svore}, \bibinfo{person}{Dave Wecker},
  {and} \bibinfo{person}{Matthias Troyer}.} \bibinfo{year}{2017}\natexlab{}.
\newblock \showarticletitle{Elucidating reaction mechanisms on quantum
  computers}.
\newblock \bibinfo{journal}{\emph{Proceedings of the National Academy of
  Sciences}}  \bibinfo{volume}{114} (\bibinfo{date}{7} \bibinfo{year}{2017}),
  \bibinfo{pages}{7555--7560}.
\newblock
Issue 29.
\showISSN{0027-8424}
\urldef\tempurl%
\url{https://doi.org/10.1073/pnas.1619152114}
\showDOI{\tempurl}


\bibitem[\protect\citeauthoryear{Ross and Selinger}{Ross and Selinger}{2016}]%
        {Ross2016Optimal}
\bibfield{author}{\bibinfo{person}{Neil~J Ross} {and} \bibinfo{person}{Peter
  Selinger}.} \bibinfo{year}{2016}\natexlab{}.
\newblock \showarticletitle{Optimal Ancilla-free Clifford+T Approximation of
  Z-rotations}.
\newblock \bibinfo{journal}{\emph{Quantum Information \& Computation}}
  \bibinfo{volume}{16} (\bibinfo{date}{9} \bibinfo{year}{2016}),
  \bibinfo{pages}{901--953}.
\newblock
Issue 11-12.
\showISSN{1533-7146}
\urldef\tempurl%
\url{http://dl.acm.org/citation.cfm?id=3179330.3179331}
\showURL{%
\tempurl}


\bibitem[\protect\citeauthoryear{Shende, Bullock, and Markov}{Shende
  et~al\mbox{.}}{2006}]%
        {shende2006synthesis}
\bibfield{author}{\bibinfo{person}{V.~V. Shende}, \bibinfo{person}{S.~S.
  Bullock}, {and} \bibinfo{person}{I.~L. Markov}.}
  \bibinfo{year}{2006}\natexlab{}.
\newblock \showarticletitle{Synthesis of quantum-logic circuits}.
\newblock \bibinfo{journal}{\emph{IEEE Transactions on Computer-Aided Design of
  Integrated Circuits and Systems}}  \bibinfo{volume}{25} (\bibinfo{date}{6}
  \bibinfo{year}{2006}), \bibinfo{pages}{1000--1010}.
\newblock
Issue 6.
\showISSN{0278-0070}
\urldef\tempurl%
\url{https://doi.org/10.1109/TCAD.2005.855930}
\showDOI{\tempurl}


\bibitem[\protect\citeauthoryear{Shor}{Shor}{1996}]%
        {Shor1996FaultTolerance}
\bibfield{author}{\bibinfo{person}{P.W. Shor}.}
  \bibinfo{year}{1996}\natexlab{}.
\newblock \showarticletitle{Fault-tolerant quantum computation}.
\newblock \bibinfo{journal}{\emph{Proceedings of 37th Conference on Foundations
  of Computer Science}}, \bibinfo{pages}{56--65}.
\newblock
\showISBNx{0-8186-7594-2}
\urldef\tempurl%
\url{https://doi.org/10.1109/SFCS.1996.548464}
\showDOI{\tempurl}


\bibitem[\protect\citeauthoryear{von Burg, Low, Häner, Steiger, Reiher,
  Roetteler, and Troyer}{von Burg et~al\mbox{.}}{2021}]%
        {vonBurg2020carbon}
\bibfield{author}{\bibinfo{person}{Vera von Burg}, \bibinfo{person}{Guang~Hao
  Low}, \bibinfo{person}{Thomas Häner}, \bibinfo{person}{Damian~S. Steiger},
  \bibinfo{person}{Markus Reiher}, \bibinfo{person}{Martin Roetteler}, {and}
  \bibinfo{person}{Matthias Troyer}.} \bibinfo{year}{2021}\natexlab{}.
\newblock \showarticletitle{Quantum computing enhanced computational
  catalysis}.
\newblock \bibinfo{journal}{\emph{Physical Review Research}}
  \bibinfo{volume}{3} (\bibinfo{date}{7} \bibinfo{year}{2021}),
  \bibinfo{pages}{033055}.
\newblock
Issue 3.
\showISSN{2643-1564}
\urldef\tempurl%
\url{https://doi.org/10.1103/PhysRevResearch.3.033055}
\showDOI{\tempurl}


\bibitem[\protect\citeauthoryear{Wang, Berry, de~Oliveira, and Sanders}{Wang
  et~al\mbox{.}}{2013}]%
        {Wang2013Solovay}
\bibfield{author}{\bibinfo{person}{Dong-Sheng Wang},
  \bibinfo{person}{Dominic~W. Berry}, \bibinfo{person}{Marcos~C. de Oliveira},
  {and} \bibinfo{person}{Barry~C. Sanders}.} \bibinfo{year}{2013}\natexlab{}.
\newblock \showarticletitle{Solovay-Kitaev Decomposition Strategy for
  Single-Qubit Channels}.
\newblock \bibinfo{journal}{\emph{Physical Review Letters}}
  \bibinfo{volume}{111} (\bibinfo{date}{9} \bibinfo{year}{2013}),
  \bibinfo{pages}{130504}.
\newblock
Issue 13.
\showISSN{0031-9007}
\urldef\tempurl%
\url{https://doi.org/10.1103/PhysRevLett.111.130504}
\showDOI{\tempurl}


\bibitem[\protect\citeauthoryear{Zeng, Cross, and Chuang}{Zeng
  et~al\mbox{.}}{2011}]%
        {Zeng2011Transversality}
\bibfield{author}{\bibinfo{person}{Bei Zeng}, \bibinfo{person}{Andrew Cross},
  {and} \bibinfo{person}{Isaac~L. Chuang}.} \bibinfo{year}{2011}\natexlab{}.
\newblock \showarticletitle{Transversality Versus Universality for Additive
  Quantum Codes}.
\newblock \bibinfo{journal}{\emph{IEEE Transactions on Information Theory}}
  \bibinfo{volume}{57} (\bibinfo{date}{9} \bibinfo{year}{2011}),
  \bibinfo{pages}{6272--6284}.
\newblock
Issue 9.
\showISSN{0018-9448}
\urldef\tempurl%
\url{https://doi.org/10.1109/TIT.2011.2161917}
\showDOI{\tempurl}


\end{thebibliography}

\appendix

\section{Concentration of randomized multiplexed rotations}

In this section, we prove the concentration of applying many randomized
instances of multiplexed rotation. Our derivation here is more-or-less
exactly the same as by Chen et al. \citep{Chen2020}, who instead
studied concentration of stochastic Trotter-based Hamiltonian simulation.
From Markov's inequality, $\mathrm{Pr}\left[|X|\ge a\right]\le\frac{\mathbb{E}\left[\varphi\left(\left|X\right|\right)\right]}{\varphi\left(\left|a\right|\right)}$
for any $a\ge0$ and monotonically increasing non-negative function
$\varphi$ satisfying $\varphi(a)>0$. Choosing $\varphi(x)=x^{q}$,
\begin{equation}
\mathrm{Pr}\left[\norm{(\tilde{Q}-Q)\ket{\psi}}\ge\epsilon\right]\le\frac{\mathbb{E}\left[\norm{(\tilde{Q}-Q)\ket{\psi}}^{q}\right]}{\epsilon^{q}}.
\end{equation}
The error is split into a systematic and a random component by adding
and subtracting its expectation
\begin{align}
\norm{(\tilde{Q}-Q)\ket{\psi}} & =\norm{\left(\tilde{Q}-\mathbb{E}\left[\tilde{Q}\right]+\mathbb{E}\left[\tilde{Q}\right]-Q\right)\ket{\psi}}\nonumber \\
 & \le\norm{\left(\mathbb{E}\left[\tilde{Q}\right]-Q\right)\ket{\psi}}+\norm{\left(\tilde{Q}-\mathbb{E}\left[\tilde{Q}\right]\right)\ket{\psi}}\nonumber \\
 & \le\underbrace{\norm{\left(\mathbb{E}\left[\tilde{Q}\right]-Q\right)}}_{\mathrm{Systematic}}+\underbrace{\norm{\left(\tilde{Q}-\mathbb{E}\left[\tilde{Q}\right]\right)\ket{\psi}}}_{\mathrm{Random}}.
\end{align}
Note that the systematic component $\norm{Q-\mathbb{E}\left[\tilde{Q}\right]}$
also upper bounds the diamond distance. As for the random component,
we use $\left(x+y\right)^{q}\le2^{q}\max\left\{ x^{q},y^{q}\right\} $
for $x,y\ge0$ to bound
\begin{equation}
\mathbb{E}\left[\norm{\left(\tilde{Q}-Q\right)\ket{\psi}}^{q}\right]\le2^{q}\max\left\{ \norm{\left(\mathbb{E}\left[\tilde{Q}\right]-Q\right)\ket{\psi}}^{q},\mathbb{E}\left[\norm{\left(\tilde{Q}-\mathbb{E}\left[\tilde{Q}\right]\right)\ket{\psi}}^{q}\right]\right\} .
\end{equation}
We now evaluate these quantities.

\subsection{Bounding the systematic shift}

The error due to the systematic component is $\norm{\left(\mathbb{E}\left[\tilde{Q}\right]-Q\right)\ket{\psi}}\le\norm{\mathbb{E}\left[\tilde{Q}\right]-Q}$.
This may be bounded by a triangle inequality on a telescoping sum
\begin{align}
\norm{\mathbb{E}\left[\tilde{Q}\right]-Q} & \le\norm{\left(\prod_{j=1}^{N-1}Q_{j}\cdot\mathbb{E}\left[\mathcal{U}_{j}\right]\right)-\left(\prod_{j=1}^{N-1}Q_{j}\cdot U\right)}\nonumber \\
 & \le N\norm{\mathbb{E}\left[\tilde{U}\right]-U}.
\end{align}
The problem thus reduces to bounding the systematic shift of a single
application of $\norm{\mathbb{E}\left[\tilde{U}\right]-U}$.

\subsection{Bounding the random shift}

We now bound the component $\mathbb{E}\left[\norm{\left(\tilde{Q}-\mathbb{E}\left[\tilde{Q}\right]\right)\ket{\psi}}^{q}\right]$
arising from fluctuations of $\tilde{Q}$ from its mean. The proof
is based on exploiting orthogonal vectors in a martingale difference
sequence using the fact 
\begin{equation}
\mathbb{E}\left[\norm{x+y}^{q}\right]^{2/q}\le\left(\mathbb{E}\norm x^{q}\right)^{2/q}+\left(q-1\right)\left(\mathbb{E}\norm y^{q}\right)^{2/q}
\end{equation}
for $q\ge2$ and random vectors $x,y$ that obey $\mathbb{E}\left[y|x\right]=0$.
Now consider a sequence of random vectors $\ket{\psi_{j}}=\left(\prod_{l=1}^{j}Q_{j}\cdot\tilde{U}_{j}\right)\cdot Q_{0}\ket{\psi}$.
Then $\left(\tilde{Q}-\mathbb{E}\left[\tilde{Q}\right]\right)\ket{\psi}=\ket{\psi_{N}}-\mathbb{E}\left[\ket{\psi_{N}}\right]$.
By adding $Q_{N}\mathbb{E}\left[\tilde{U}_{N}\right]\ket{\psi_{N-1}}-Q_{N}\mathbb{E}\left[\tilde{U}_{N}\right]\ket{\psi_{N-1}}=0$,
\begin{equation}
\left(\tilde{Q}-\mathbb{E}\left[\tilde{Q}\right]\right)\ket{\psi}=\underbrace{Q_{N}\left(\tilde{U}_{N}-\mathbb{E}\left[\tilde{U}_{N}\right]\right)\ket{\psi_{N-1}}}_{y}+\underbrace{Q_{N}\mathbb{E}\left[\tilde{U}_{N}\right]\left(\ket{\psi_{N-1}}-\mathbb{E}\left[\ket{\psi_{N-1}}\right]\right)}_{x}.
\end{equation}
Observe that conditioning on $x$ sets the value of $\ket{\psi_{N-1}}$
but not $\tilde{U}_{N}$. Hence $\mathbb{E}\left[Q_{N}\left(\tilde{U}_{N}-\mathbb{E}\left[\tilde{U}_{N}\right]\right)\ket{\psi_{N-1}}|x\right]=Q_{N}\mathbb{E}\left[\left(\tilde{U}_{N}-\mathbb{E}\left[\tilde{U}_{N}\right]\right)\right]\ket{\psi_{N-1}}=0$.
and we may apply inequality on $\mathbb{E}\left[\norm{x+y}^{q}\right]^{2/q}$
to obtain
\begin{align}
\mathbb{E}\left[\norm{\ket{\psi_{N}}-\mathbb{E}\left[\ket{\psi_{N}}\right]}^{q}\right]^{2/q} & =\mathbb{E}\left[\norm{\left(\tilde{Q}-\mathbb{E}\left[\tilde{Q}\right]\right)\ket{\psi}}^{q}\right]^{2/q}=\mathbb{E}\left[\norm{x+y}^{q}\right]^{2/q}\nonumber \\
 & \le\mathbb{E}\left[\norm{\mathbb{E}\left[\tilde{U}_{N}\right]\left(\ket{\psi_{N-1}}-\mathbb{E}\left[\ket{\psi_{N-1}}\right]\right)}^{q}\right]^{2/q}\nonumber \\
 & \qquad\qquad+\left(q-1\right)\mathbb{E}\left[\norm{\left(\tilde{U}_{N}-\mathbb{E}\left[\tilde{U}_{N}\right]\right)\ket{\psi_{N-1}}}^{q}\right]^{2/q}\nonumber \\
 & \le\mathbb{E}\left[\norm{\left(\ket{\psi_{N-1}}-\mathbb{E}\left[\ket{\psi_{N-1}}\right]\right)}^{q}\right]^{2/q}+\left(q-1\right)\mathbb{E}\left[\norm{\left(\tilde{U}_{N}-\mathbb{E}\left[\tilde{U}_{N}\right]\right)}^{q}\right]^{2/q}\nonumber \\
 & \le\left(q-1\right)\sum_{l=1}^{N}\mathbb{E}\left[\norm{\left(\tilde{U}_{l}-\mathbb{E}\left[\tilde{U}_{l}\right]\right)}^{q}\right]^{2/q}
\end{align}
We now substitute $\norm{\left(\tilde{U}_{l}-\mathbb{E}\left[\tilde{U}_{l}\right]\right)}\le\frac{\pi}{2^{-b+1}}$
to obtain
\begin{equation}
\mathbb{E}\left[\norm{\left(\tilde{Q}-\mathbb{E}\left[\tilde{Q}\right]\right)\ket{\psi}}^{q}\right]^{1/q}\le2\pi\sqrt{\frac{\left(q-1\right)N}{2^{2b}}}.
\end{equation}

\subsection{Tail bounds on approximation error}

We now combine the bounds on systematic and random errors.

\begin{align}
\mathbb{E}\left[\norm{\left(\mathcal{\tilde{Q}}-Q\right)\ket{\psi}}^{q}\right] & \le2^{q}\max\left\{ \norm{\left(\mathbb{E}\left[\tilde{Q}\right]-Q\right)\ket{\psi}}^{q},\mathbb{E}\left[\norm{\left(\tilde{Q}-\mathbb{E}\left[\mathcal{Q}\right]\right)\ket{\psi}}^{q}\right]\right\} \nonumber \\
 & \le2^{q}\max\left\{ \left(\frac{\pi^{2}}{2}\frac{N}{2^{2b}}\right)^{q},\left(2\pi\sqrt{\frac{\left(q-1\right)N}{2^{2b}}}\right)^{q}\right\} \nonumber \\
 & \le\left(4\pi\sqrt{\frac{\left(q-1\right)N}{2^{2b}}}\right)^{q}\le\left(\frac{4\pi}{2^{b}}\sqrt{qN}\right)^{q}.
\end{align}
Using Lyapunov's inequality, $\mathbb{E}\left[|X|^{q}\right]^{1/q}\le\mathbb{E}\left[|X|^{s}\right]^{1/s}$
for $0<q<s<\infty.$ Hence the expected error
\begin{align}
\bar{\epsilon}=\max_{\ket{\psi}}\mathbb{E}\left[\norm{\left(\tilde{Q}-Q\right)\ket{\psi}}\right] & \le\max_{\ket{\psi}}\mathbb{E}\left[\norm{\left(\tilde{Q}-Q\right)\ket{\psi}}^{2}\right]^{1/2}\le\frac{4\pi}{2^{b}}\sqrt{N}.\\
b & \ge\log_{2}\left(\frac{4\pi\sqrt{N}}{\bar{\epsilon}}\right)\approx3.65+\log_{2}\left(\frac{\sqrt{N}}{\bar{\epsilon}}\right)
\end{align}
Substituting into Markov's inequality, 
\begin{align}
\mathrm{Pr}\left[\norm{(\mathcal{\tilde{Q}}-Q)\ket{\psi}}\ge\epsilon\right] & \le\frac{\mathbb{E}\left[\norm{(\mathcal{\tilde{Q}}-Q)\ket{\psi}}^{q}\right]}{\epsilon^{q}}\nonumber \\
 & \le\left(\frac{4\pi}{\epsilon2^{b}}\sqrt{qN}\right)^{q}\nonumber \\
 & =\exp\left(q\left(\frac{1}{2}\ln\left(q\right)+\ln\left(\frac{4\pi}{\epsilon2^{b}}\sqrt{N}\right)\right)\right)\nonumber \\
 & =\exp\left(q\left(\frac{1}{2}\ln\left(q\right)+\ln\left(z\right)\right)\right)
\end{align}
The exponent is minimized by choosing $q=\frac{1}{ez^{2}}=\frac{\epsilon^{2}2^{2b}}{16e\pi^{2}N}$.
Hence
\begin{align}
p=\mathrm{Pr}\left[\norm{(\tilde{Q}-Q)\ket{\psi}}\ge\epsilon\right] & \le\exp\left(q\left(\ln\left(\sqrt{c}\right)\right)\right)=\exp\left(-\frac{\epsilon^{2}2^{2b}}{32e\pi^{2}N}\right)
\end{align}
The bits of precision required is then
\begin{align}
b & \ge\frac{1}{2}\log_{2}\left(32e\pi^{2}\log\left(\frac{1}{p}\right)\right)+\log_{2}\left(\frac{\sqrt{N}}{\epsilon}\right)\nonumber \\
 & \ge4.88+\log_{2}\left(\log\left(\frac{1}{p}\right)\right)+\log_{2}\left(\frac{\sqrt{N}}{\epsilon}\right).
\end{align}

\end{document}